\documentclass[letterpaper, 10 pt, conference]{ieeeconf}

\IEEEoverridecommandlockouts
\usepackage{graphics} 
\usepackage{epsfig} 
\usepackage{mathptmx} 
\usepackage{times} 
\usepackage{amsmath} 
\usepackage{mathtools}
\usepackage{amsthm}
\usepackage{amssymb}  
\usepackage{twoopt}
\usepackage[backend=biber,style=numeric,sorting=none]{biblatex}
\usepackage{xfrac}
\usepackage{subcaption}
\usepackage{float}
\usepackage{tikz}
\usepackage{pgfplots}
\pgfplotsset{width=6cm,compat=1.12}
\usepgfplotslibrary{fillbetween}

\title{\LARGE \bf LQR Control with Sparse Adversarial Disturbances}

\author{Samuel Pfrommer and Somayeh Sojoudi \thanks{This work was supported by grants from ONR, NSF and C3.ai Digital Transformation Institute. Samuel Pfrommer and Somayeh Sojoudi are with the Department of Electrical Engineering and Computer Sciences at the University of California, Berkeley. {\tt Email: sam.pfrommer@berkeley.edu, sojoudi@berkeley.edu}}}

\newcommand{\R}{{\mathbb R}}
\newcommand{\Z}{{\mathbb Z}}
\newcommand{\Zp}{{\mathbb Z}_{+}}

\newcommand{\I}{I}

\DeclareMathOperator*{\E}{\mathbb{E}}

\newcommand{\T}{^\intercal}

\newcommand{\defeq}{\vcentcolon=}
\newcommand{\eigmin}[1]{\lambda_{\min}(#1)}
\newcommand{\isep}{\mathrel{{.}\,{.}}\nobreak}
\newcommand{\bic}{[\,}
\newcommand{\eic}{\,]}
\newcommand{\bio}{(\,}
\newcommand{\eio}{\,)}

\newcommand{\J}[1][k]{{J_t^{#1}}}
\newcommand{\Jp}[1][k]{{J_{t+1}^{#1}}}

\renewcommand{\P}[1][t]{P_{#1}}
\newcommand{\Pp}{P_{t+1}}
\newcommand{\Pb}{\hat{P}}
\newcommand{\K}[1][t]{{K_{#1}}}

\newcommand{\F}{{F_t}}

\newcommand{\Ln}[1][k]{{r_t^{#1}}}
\newcommand{\Le}{r}
\newcommand{\Lp}[1][k]{{r_{t+1}^{#1}}}
\newcommand{\Lt}[1][k]{{\tilde{r}_t^{#1}}}
\newcommand{\Lpt}[1][k]{{\tilde{r}_{t+1}^{#1}}}
\newcommand{\Li}[1][k_i]{{r_{t_i}^{#1}}}
\newcommand{\Lj}[1][k_j']{{r_{t_j'}^{#1}}}
\newcommand{\Lpi}[1][k_i]{{r_{t_i+1}^{#1}}}

\newcommand{\Lpti}[1][k_i]{{\tilde{r}_{t_i+1}^{#1}}}

\newcommand{\C}[1][k]{{c_t^{#1}}}

\newcommand{\Ct}[1][k]{{\tilde{c}_t^{#1}}}
\newcommand{\Cpt}[1][k]{{\tilde{c}_{t+1}^{#1}}}

\renewcommand{\Pr}[1][k]{{p_t^{#1}}}
\newcommand{\Pre}{p}
\newcommand{\Pri}[1][k_i]{{p_{t_i}^{#1}}}

\newcommand{\Prb}[1][k]{{\overline{p}_t^{#1}}}
\newcommand{\Prib}[1][k_i]{{\overline{p}_{t_i}^{#1}}}

\newcommand{\w}{w_k}
\newcommand{\wb}{\hat w}

\newcommand{\xkp}{x_{t_k'}}
\newcommand{\ykp}{y_{t_k'}}
\newcommand{\Pkp}{P_{t_k'}}

\theoremstyle{customstyle}
\newtheorem{theorem}{Theorem}

\newtheorem{assumption}{Assumption}
\newtheorem{lemma}{Lemma}
\newtheorem{remark}{Remark}

\binoppenalty=\maxdimen
\relpenalty=\maxdimen

\begin{document}

\maketitle
\thispagestyle{empty}
\pagestyle{empty}

\begin{abstract}
    Recent developments in cyber-physical systems and event-triggered control have led to an increased interest in the impact of sparse disturbances on dynamical processes. We study Linear Quadratic Regulator (LQR) control under sparse disturbances by analyzing three distinct policies: the blind online policy, the disturbance-aware policy, and the optimal offline policy. We derive the two-dimensional recurrence structure of the optimal disturbance-aware policy, under the assumption that the controller has information about future disturbance values with only a probabilistic model of their locations in time. Under mild conditions, we show that the disturbance-aware policy converges to the blind online policy if the number of disturbances grows sublinearly in the time horizon. Finally, we provide a finite-horizon regret bound between the blind online policy and optimal offline policy, which is proven to be quadratic in the number of disturbances and in their magnitude. This provides a useful characterization of the suboptimality of a standard LQR controller when confronted with unexpected sparse perturbations.
\end{abstract}

\section{Introduction}
Much of the recent focus on sparse disturbances is motivated by the study of cyber-physical systems, which couple physical processes, networking, and computation into a larger controller architecture. Research in this area seeks to provide robustness against malicious cyberattacks and Stuxnet-style adversaries. Here, adversaries often act sparsely in both measurement channels and time in order to avoid detection and efficiently use their limited computational resources. Early works on this topic focused largely on state estimation where adversaries can perturb a sparse subset of measurement channels \cite{fawzi2014secure, lu2017secure, Yong2016RobustAR, Mo2016OnTP}. In \cite{Anubi2018RobustRS}, the authors examined the signal reconstruction problem for cyber-physical systems with limited knowledge of the support of the attack vector. This was further explored in \cite{xie2017secure}, which studied state estimation in the context of sparse actuator attacks. Sparsity in time was first introduced by \cite{jingyang2016sparse}, which restricted attacks to a small sequence of consecutive time steps and considered the adversarial goal of maximizing the system estimation error of a Kalman filter. We generalize this notion of sparsity to nonconsecutive time steps.

Sparsity has also gained significant attention in the realm of event-triggered control (ETC). Many cyber-physical systems face communication bandwidth constraints and actuation limits, inspiring event-triggered techniques that provide sparse-in-time control signals \cite{banno2021data, gommans2015resource, Li2018ModelBasedAE}. Model-based ETC instead addresses the issue of sparse disturbances with time-varying magnitudes \cite{brunner2016dynamic}. In model-based ETC, both the sensor and the actuator maintain an internal model of the system, and the sensor only communicates to the actuator if it is triggered by a large deviation from the nominal state. Numerical experiments in \cite{brunner2016dynamic} demonstrate that this approach is particularly sample-effective when dealing with sparse disturbances.

Fault detection is a related field which aims to detect and predict sparse anomalies in system operation. Early work focused on linear systems, using techniques such as detection filters \cite{white1987detection} and parameter estimation \cite{isermann1984process}. More recently, \cite{sanz2016enhanced} considers non-sparse disturbance prediction for continuous-time systems by computing higher order derivatives of the disturbance signal. Interesting related work on distributed parameter systems employed detection residuals to estimate fault functions and time to failure \cite{cai2016model}. Learning-based approaches have also been considered for model-unknown nonlinear systems \cite{zhengdao2008new}. Our novel disturbance-aware policy assumes access to such probabilistic predictions of future disturbances. Modeling and estimating future disturbances is outside of the scope of this work and would rely on techniques similar to those in the fault detection literature.

Other recent research efforts have focused on the role of sparsity in identification of dynamical systems. Inspired by modeling challenges in large-scale control problems, \cite{feng2021learning} considers the identification of linear time-invariant systems under large and sparse disturbances. These disturbances can model both stealthy adversarial attacks and faults in system operation. The authors introduce sparsity via the $\Delta$-spaced disturbance model, which assumes that adversaries are unable to perturb the system twice within $\Delta$ time steps. Under this assumption, it is shown that perfect recovery of the system matrix is possible under certain sufficient conditions. This work builds off of \cite{fattahi2018data} and \cite{fattahi2018sample}, which examine system identification when the dynamics matrices themselves have a sparse structure and provide sample complexity guarantees for accurate system estimation.

While this is the first work to explicitly examine sparse disturbances in an LQR context, we draw inspiration from several recent advancements in the area of policy regret. Leveraging new techniques from the online convex optimization community, \cite{agarwal2019online} formulates a regret minimization problem in the case of a linear dynamical system with convex cost functions and arbitrary bounded disturbances. Their proposed algorithm achieves $O(\sqrt{T})$ regret against the optimal stable linear policy. This is extended in \cite{Hazan2020TheNC}, which considers the case of unknown dynamics matrices and provides an algorithm with an $O(T^{2/3})$ regret bound. \cite{goel2020power} reimposes the LQR condition of quadratic cost functions and fully characterizes the recursive structure of the optimal offline policy under this framework. Furthermore, they show that the time-averaged regret of the blind online policy with respect to the optimal offline policy grows linearly in the time horizon. We extend these results to the sparse-disturbance setting with a bound that is independent of the time horizon and varies quadratically with the number of disturbances and their magnitude. Our complete contributions are summarized in \ref{sec:contributions}.

\subsection{Notation}
The spectral norm of a matrix is denoted by $\rho(\cdot)$, and the $l_2$ norm of a vector is denoted by $\| \cdot \|$. We let $M \succ 0$ and $M \succeq 0$ indicate that $M$ is positive definite and positive semidefinite, respectively. The minimum eigenvalue of a matrix is given by $\eigmin{\cdot}$. An upper bound on the norm of a collection of matrices or vectors is generally denoted using a hat; e.g., $\hat w \geq \| \w \| ~ \forall k$. We denote the set of integers by $\Z$, the set of nonnegative integers by $\Zp$, the set of real numbers by $\R$, and the expectation of a random variable by $\E$. The set of integers between $a \in \Z$ and $b \in \Z$ inclusive is denoted by $\bic a \isep b \eic$, with $\bio a \isep b \eio$ corresponding to the open integer interval omitting $a$ and $b$.

\section{Problem Statement} \label{sec:probstatement} 
We consider the control of a linear time-invariant system that evolves according to
\begin{align*}
    x_{t+1} = A x_t + B u_t + d_t, \quad t = 0, 1, \dots, T-1,
\end{align*}
where $A \in \R^{n \times n}$ and $B \in \R^{n \times m}$ are system matrices, $x_t \in \R^n$ is the system state, $u_t \in \R^m$ is the control input, $d_t \in \R^n$ is an external disturbance, and $T$ is the time horizon. Moreover, we focus on the scenario where the disturbances $d_t$ are \textit{sparse}; i.e., $d_t \neq 0$ if and only if $t \in D$, for some index set $D = \{ t_1, t_2, \dots t_{|D|}\}$, where $|D|$ denotes the cardinality of $D$. The nonzero disturbances are given by $d_{t_k} = \w$, $k \in 1, \dots, |D|$, and are bounded in $l_2$ norm by $\wb \geq \| \w \| ~ \forall k$.

The objective of our policies is to minimize the cost
\begin{align} \label{eqn:cost}
    \E_{d_0, \dots d_{T-1}} \Big( x_T\T Q_T x_T + \sum_{t=0}^{T-1} \left[ x_t\T Q x_t + u_t\T R u_t \right] \Big),
\end{align}
where $Q_T, Q \in \R^{n\times n}$ are positive semidefinite and $R \in \R^{m \times m}$ is positive definite. We compare the following three policies in this work:
\begin{enumerate}
    \item The \textbf{blind online policy} executes standard LQR control assuming no disturbances; i.e., assuming $d_t = 0$ for all $t$, even though this assumption may not hold. This is the same policy that would be executed under the typical condition that disturbances $d_t$ are independent with zero mean. The optimal disturbance-free policy is then a linear controller of the form ${u_t = -K_t x_t}$, for $t=\bic 0 \isep T \eio$, where 
    \begin{align} 
        K_t &= (B\T \Pp B + R)^{-1}B\T \Pp A, \label{eqn:feedback}
    \end{align}
    and the sequence $\P$ arises from solving the discrete-time Riccati equation
    \begin{align} 
        \mathclap{\P = A\T (\Pp - \Pp B (B\T \Pp B + R)^{-1} B\T \Pp) A + Q,\qquad} \label{eqn:riccati}
    \end{align}
    with $\P[T] = Q_T$ \cite{bertsekas2012dynamic, kalman1960contributions}.
\item The \textbf{disturbance-aware policy} has knowledge of the disturbances $\w$ and their ordering but is unaware of their exact locations in time. More specifically, at time step $t$, the policy is aware of potential future disturbances $w_{k'}, \dots, w_{|D|}$, where $k'$ is the smallest $k$ such that $t_{k} > t$. Instead of knowing the precise disturbance locations $t_k$, the policy is given access to a probabilistic model of a disturbance occurring at time step $t$ when there are $k$ disturbances remaining. We denote this probability as $\Pr$. One simple probabilistic model is the uniform conditional model $p_t^k = k/(T - t)$, which assumes that the remaining disturbances are distributed uniformly over the time horizon. The form of the policy is derived explicitly in Section~\ref{sec:distawarestruct}.
    \item The \textbf{optimal offline policy} has complete knowledge of all disturbances $d_t$. Under this setting, the optimal control and cost have the following structure, first derived in \cite{goel2020power}:
    \begin{subequations}
    \begin{align}
        u_t^* &= - K_t x_t - (B\T \Pp B + R)^{-1} B\T (\Pp d_t + \sfrac{1}{2} ~ v_{t+1}), \\
        V_t(x) &= x\T \P x + v_t\T x + q_t, \label{eqn:optofflinecost} \\[10pt]
        v_t &= 2 A\T S_t d_t + A\T S_t \Pp^{-1} v_{t+1}, \label{eqn:optofflinevector}\\
        q_t &= q_{t+1} + d_t \T S_{t+1} d_t + v_{t+1}\T \Pp^{-1} S_t d_t \nonumber \\
            & \quad - \sfrac{1}{4} ~ v_{t+1}\T B (B\T \Pp B + R)^{-1} B\T v_{t+1}, \label{eqn:optofflinescalar}\\
        S_t &\defeq \Pp - \Pp B (B\T \Pp B + R)^{-1} B\T \Pp, \label{eqn:optofflineS}
    \end{align}
    \end{subequations}
    where $V_t(x)$ represents the cost-to-go of a state $x$ at time $t$, and $v_t \in \R^n$, $q_t \in \R$ are recurrences that depend only on the noise $d_t$ and are initialized to $v_T = 0$, $q_T = 0$. We note that the optimal control is determined by the optimal online policy action and a counterfactual second term depending exclusively on future disturbances.
\end{enumerate}

\subsection{Contributions} \label{sec:contributions}
We begin by deriving the optimal structure for the disturbance-aware policy in Section~\ref{sec:distawarestruct}, which involves a two-dimensional recurrence over the time step $t$ and the number of remaining disturbances $k$. This holds for any arbitrary disturbance sequence. Unlike the optimal offline policy originally considered by \cite{goel2020power}, we only assume probabilistic knowledge about the possibility of a disturbance. In practice, future disturbances and their probabilities could be inferred from data using techniques similar to the fault prediction literature, although this is outside the scope of our work.

Section~\ref{sec:distawareconv} establishes that under mild conditions, the disturbance-aware policy converges to the blind online policy in the infinite-horizon case if the number of disturbances grows sublinearly in the time horizon. This formalizes the intuition that as disturbances become more sparse, without exact knowledge of disturbance times, the disturbance-aware policy eventually takes the same actions as the blind policy. We note that this is not the case for the optimal offline policy, which can leverage its precise knowledge of disturbance times to modify its control strategy. 

Section~\ref{sec:blindregret} provides a regret bound between the blind online policy and the optimal offline policy. Specifically, we show that the regret grows quadratically in the number of disturbances and in their magnitude. We accomplish this via an intermediate reduction to the cost of the blind online policy without disturbances, which is also of independent interest. This bound is relevant in a practical setting where naive LQR controllers are deployed in a real-world environment with sparse disturbances.

\section{The Disturbance-aware Policy} \label{sec:distaware}

We derive the structure of the disturbance-aware policy, which takes the form of a two-dimensional recursion. We then show that under a particular sparsity condition the optimal control action converges to that of the blind online policy in the infinite-horizon case. For notational convenience, in this section disturbances $d_{t_k} = \w$ are indexed in reverse chronological order, such that $t_1 \geq t_2 \geq \cdots \geq t_{|D|}$. 
\subsection{Disturbance-aware policy structure} \label{sec:distawarestruct}
\begin{theorem} The form of the optimal disturbance-aware policy described in Section~\ref{sec:probstatement} is given by
\begin{align}
    (u_t^k)^* &= -\K x_t -(R + B\T \Pp B)^{-1} B\T (\Pr \Pp \w + \Lpt), \nonumber\\
    \J(x) &= x \T \P x + 2 \Ln\T x + \C, \nonumber \\[10pt]
    \Ln &= A\T (I - \F \Pp)\T (\Lpt + \Pr \Pp \w), \nonumber \\
    \C &= \Cpt - (\Lpt + \Pr \Pp \w)\T \F (\Lpt + \Pr \Pp \w) \label{eqn:distawarecontrol} \\
       & \qquad + \Pr \w\T \Pp \w + 2 \Pr \Lp[k-1] \w, \nonumber \\[10pt] 
    \F &\defeq B (R + B\T \Pp B)^{-1} B\T, \nonumber \\
    \Lt &= \Prb \Ln + \Pr \Ln[k-1], \nonumber \\
    \Ct &= \Prb \C + \Pr \C[k-1], \nonumber  
\end{align}
where $\P$ is obtained from $\Pp$ via the Riccati sequence \eqref{eqn:riccati} and $\Ln \in \R^{n}$ and $\C \in \R$ are recurrences indexed by both $t \in \bic 0 \isep T \eic$ and $k \in \bic 0 \isep |D| \eic$, initialized to $P_{T}=Q_T$, $r_{T}^k=0$, and $c_{T}^k=0$ for all $k$. Here $(u_t^k)^*$ is the expectation-optimal action at time step $t$ with $k$ possible future disturbances.
\end{theorem}
\begin{proof}[Sketch of proof]
The derivation of the optimal disturbance-aware policy mirrors the classic dynamic programming approach \cite{bertsekas2012dynamic}. Our induction hypothesis is that the cost at time step $t+1$ with $k$ possible future disturbances is of the form
\begin{align*}
    \Jp(x) = x\T \P x + 2 \Ln\T x + \C.
\end{align*}
We proceed backwards by induction, optimizing for the expected cost-minimizing control
\begin{align*}
    \J(x) = \min_u \quad \big\{ &x\T Q x + u\T R u \\
    + &\Prb \Jp(A x + B u) + \Pr \Jp[k-1] (Ax + B u + \w) \big\},
\end{align*}
where $\Prb = 1 - \Pr$. The last two terms represent the expected value of the cost-to-go, which is decomposed over the probability of a disturbance at that time step. Since we assume access to the probabilistic model $\Pr$ and disturbances $\w$, this recursion can be evaluated. Paralleling derivations in \cite{goel2020power}, we expand $\Jp$ and $\Jp[k-1]$ and collect terms to write
\begin{align*}
    \J(x) = \min_u
    \begin{bmatrix} u \\ a \end{bmatrix} \T
    \begin{bmatrix}
        R + B^T \Pp B & z \\
        z\T & M
    \end{bmatrix}
    \begin{bmatrix} u \\ a \end{bmatrix}
    + \Cpt,
\end{align*}
with the following additional definitions:
\begin{align*}
    a &= \begin{bmatrix} x\T & \w\T & \Lp\T & \Lp[k-1] \T \end{bmatrix} \T, \\
    z &= \begin{bmatrix} B \T \Pp A & \Pr B \T \Pp & \Prb B \T & \Pr B \T \end{bmatrix}, \\
    M &=
    \begin{bmatrix}
        Q + A\T \Pp A & A\T \Pp & \Prb A \T & \Pr A \T \\
        \Pp A & \Pr \Pp & 0 & \Pr \I \\
        \Prb A & 0 & 0 & 0 \\
        \Pr A & \Pr \I & 0 & 0
    \end{bmatrix}.
\end{align*}
We now apply the Schur complement to write the optimal control and cost in closed form, attaining the form specified in the theorem statement.
\end{proof}
As in \cite{goel2020power}, the optimal control is affine in $x$ with the same feedback matrix \eqref{eqn:feedback} as the blind online policy and an additional term that depends on future disturbances and the probability model. Note that $(u_t^k)^*$ now depends on both $k$ and $t$.

\subsection{Convergence to the blind online policy} \label{sec:distawareconv}
We show that under mild conditions on the system and probability model, the optimal disturbance-aware policy derived in the previous section converges to the blind online policy if the number of disturbances grows sublinearly in the time horizon.
\begin{assumption} \label{ass:controllability}
    The pair $(A,B)$ is \textit{controllable} and $Q$ can be decomposed as $Q = L\T L$ for some matrix $L$ such that $(A,L)$ is \textit{observable}. Moreover, the Riccati sequence $\P$ is stabilizing such that
    \[
        \rho(A - B \K) < 1 - \gamma
    \]
    for all $t$ and some $\gamma > 0$, where $\K$ is the feedback matrix \eqref{eqn:feedback} in the standard LQR controller.
\end{assumption}
Controllability and observability are standard assumptions in the classic LQR derivation and are sufficient to show that the controller converges to a stabilizing solution. Assumption~\ref{ass:controllability} also stipulates that intermediate iterations of the Riccati equation produce stabilizing controllers. This is immediately satisfied for many choices of the terminal cost $Q_T$ \cite{DeSouza1989Stabilizing, Bitmead1985MonotonicityAS}. Additionally, the Riccati sequence is known to converge exponentially to a stable solution \cite{kailath2000linear}, and hence the number of intermediate iterations violating Assumption~\ref{ass:controllability} is generally negligible. The positive constant $\gamma$ simply quantifies the level of stability. 

\begin{assumption} \label{ass:decay}
    The disturbance probability $\Pr$ approaches zero if the number of disturbances $k$ grows sublinearly in $t$. Formally, let $k_i$ be a monotonically increasing sequence in $\Zp$ and $t_i$ be a monotonically strictly decreasing sequence in $\bio -\infty \isep T \eic$. Then if $k_i / t_i \to 0$ as $i \to \infty$, the disturbance probability $\Pri \to 0$.
\end{assumption}

Note that in order to show infinite-horizon convergence of the disturbance-aware policy, we adopt a slight abuse of notation by letting $t \to -\infty$ with a fixed final time $T$. This is a standard approach in LQR derivations; see \cite{bertsekas2012dynamic} for reference. Assumption \ref{ass:decay} is satisfied by most natural disturbance priors. For instance, the uniform distribution $\Pr = k / (T - t)$ decays to zero if $t \to -\infty$ faster than $k \to \infty$. We now introduce a lemma before presenting the main result of this section.

\begin{lemma} \label{lem:bounded}
    Under Assumption~\ref{ass:controllability}, we have that $\| \Ln \| \leq \hat r ~$ for all $k \in \Zp$ and $t \in \bio -\infty \isep T \eic$, where $\hat r$ is some scalar.
\end{lemma}
\begin{proof}
    The recurrence for $\Ln$ can be written as
    \begin{align}
        \Ln = U_t \Lpt + \Pr H_t \w, \label{eqn:Lrecurrence}
    \end{align}
    with
    \begin{align*}
        U_t &= A\T (I - \F \Pp)\T = (A - B \K)\T, \\
        H_t &= A\T (I - \F \Pp)\T \Pp.
    \end{align*}
    Note that $U_t$ and $H_t$ depend only on $\P$. By Assumption~\ref{ass:controllability}, $\P$ converges to a positive definite solution and therefore $U_t \to U^*$ and $H_t \to H^*$ where $\|U_t\| < 1 - \gamma$ and $\|U^*\| < 1 - \gamma$.
    
    We now prove the claim by upper bounding a sequence of upper bounds given by $\hat r_t \geq \| \Ln \| ~ \forall k$. Since $r_{T}^k=0$ for all $k$, we have that $\hat r_T = 0$. By induction,
    \begin{align*}
        \| \Ln \| &\leq \| \Lpt \| \| U_t \| + \Pr \| \w \| \| H_t \| \\
                 &\leq \hat r_{t+1} (1 - \gamma) + \wb \hat H \\
                 &= \hat r_t,
    \end{align*}
    where $\hat r_t \defeq \hat r_{t+1} (1 - \gamma) + \wb \hat H$, $\wb \geq \| \w \|$ for all $k$ by assumption and $\hat H \geq \| H_t \|$ for all $t$ by convergence. Since $\gamma > 0$, $\hat{r_t}$ is a geometrically decaying sequence with a bounded additive term which cannot diverge. Hence, $\hat r_t < \hat r ~ \forall t$ for some $\hat r$ and therefore $\| \Ln \| < \hat r$ for all $k$ and $t$.
\end{proof}

\begin{figure}
    \includegraphics[width=\linewidth]{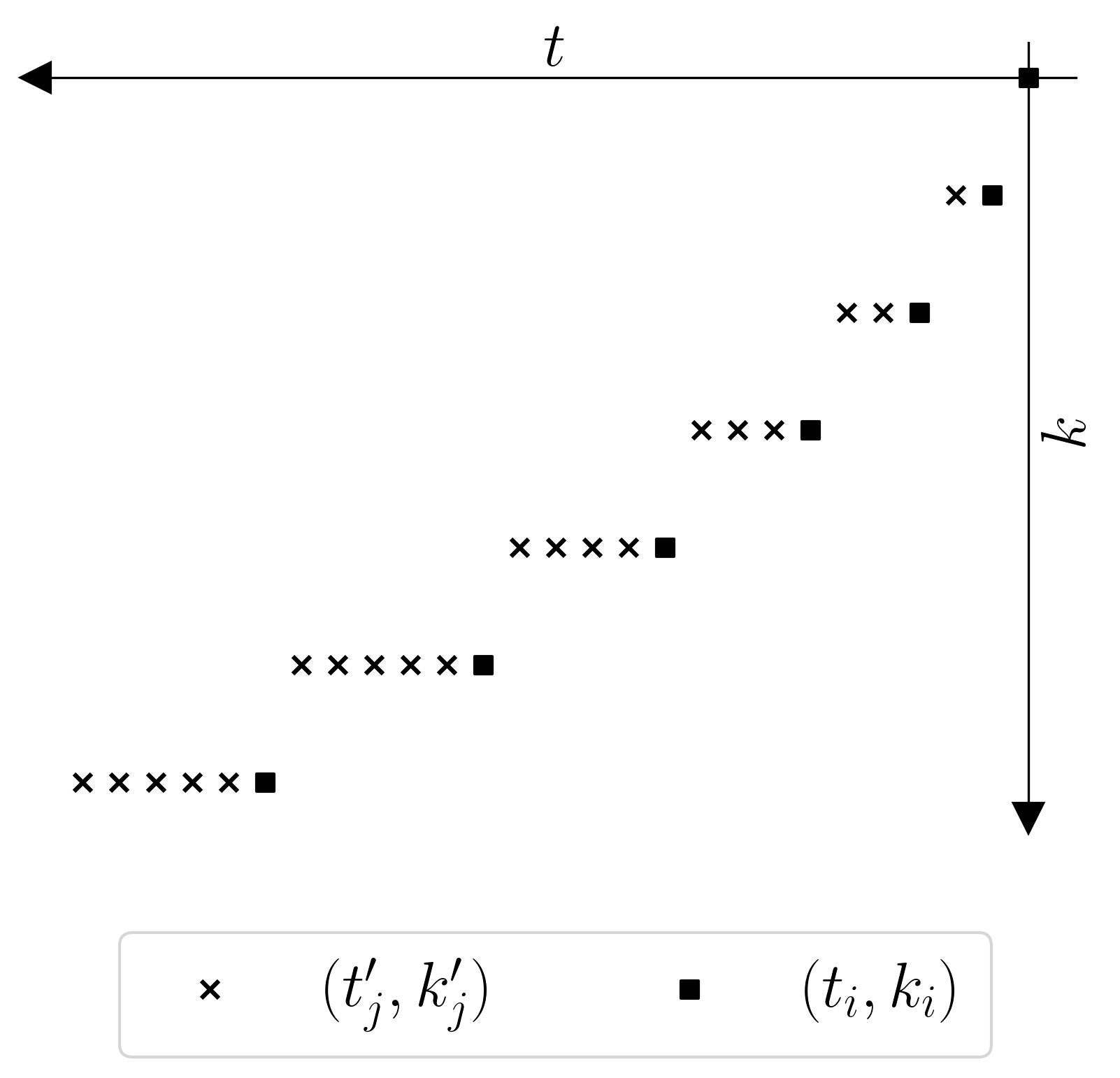}
    \caption{
        An illustration of a trajectory disturbance budget $k$ growing sublinearly in the time horizon. An increase in $k$ corresponds to a nonzero disturbance in the system dynamics. Lemma~\ref{lem:subseqconverge} shows that along the tuples $(t_i, k_i)$ demarcated by black squares, the norm of the linear cost coefficient $\| r_{t_i}^{k_i} \|$ vanishes. Theorem~\ref{thm:distawareconv} extends this to the intermediate values $(t_j', k_j')$, and establishes that the disturbance-aware policy converges to the blind online policy. 
    }
    \label{fig:distsublin}
\end{figure}

\begin{lemma} \label{lem:subseqconverge}
    Under Assumptions~\ref{ass:controllability}~and~\ref{ass:decay}, let $k_i / t_i \to 0$ as $i \to \infty$ for strictly monotonic sequences $k_i$ and $t_i$. Since only one disturbance can occur on a particular time step, we further assume $k_i - k_{i-1} \in \{0, 1\} ~ \forall i$. Then $\| \Li \| \to 0$ as $i \to \infty$.
\end{lemma}
\begin{proof}
    Let $\bar t_i = t_{i+1} - t_i$. Using \eqref{eqn:Lrecurrence} and expanding $\Lpti$ yields
    \begin{align}
        \Li = \Prib U_{t_i} \Lpi + \Pri U_{t_i} \Lpi[k_i-1] + \Pri H_{t_i} w_{k_i}. \label{eqn:Lexpansion}
    \end{align}
    We focus on bounding our two terms $\Lpi$ and $\Lpi[k_i-1]$. By inspection of \eqref{eqn:Lrecurrence},
    \begin{align*}
        \| \Lpi \| &\leq (1 - \gamma)^{\bar t_i - 1} \hat r + \sum_{j=1}^{\bar t_i-1} \Pre_{t_i + j}^{k_i} (1 - \gamma)^j (\hat r + \wb \hat H) \\
                   &\leq (1 - \gamma)^{\bar t_i - 1} \hat r + \hat \Pre_i \cdot \frac{\hat r + \wb \hat H}{\gamma},
    \end{align*}
    where $\wb \geq \| \w \|$ upper bounds the disturbance magnitude, $\hat H$ upper bounds $\| H_t \|$, $\hat r$ is as in Lemma~\ref{lem:bounded}, and
    \[
        \hat \Pre_i = \max_{j=1, \dots, \bar t_i - 1} \{\Pre_{t_i + j}^{k_i}\}.
    \]
    It is evident that under Assumption~\ref{ass:decay} with $k_i/t_i \to 0$ that $\hat \Pre_i \to 0$ as well. Note that we leverage the observation that $\| U_t \| < 1 - \gamma$ as in the proof for Lemma~\ref{lem:bounded}. We similarly bound
    \begin{align*}
        \| \Le_{t_i+1}^{k_i - 1} \| &\leq (1 - \gamma)^{\bar t_i - 1} \| \Le_{t_{i-1}}^{k_{i-1}} \|
            + \sum_{j=1}^{\bar t_i-1} \Pre_{t_i + j}^{k_i-1} (1 - \gamma)^j (\hat r + \wb \hat H) \\
                   &\leq (1 - \gamma)^{\bar t_i - 1} \| \Le_{t_{i-1}}^{k_{i-1}} \| + \hat \Pre_i' \cdot \frac{\hat r + \wb \hat H}{\gamma},
    \end{align*}
    where
    \[
        \hat \Pre_i' = \max_{j=1, \dots, \bar t_i - 1} \{\Pre_{t_i + j}^{k_i-1}\}.
    \] 
    Similarly to $\hat \Pre_i$, as $i \to \infty$ we have $\hat \Pre_i' \to 0$. We can now bound the expression for $\Li$ in \eqref{eqn:Lexpansion}:
    \[
        \| \Li \| \leq (1 - \gamma)^{\bar t_i} \| \Le_{t_{i-1}}^{k_{i-1}} \| + (\hat \Pre_i + \hat \Pre_i') \frac{\hat r + \wb \hat H}{\gamma} + \Pri \wb \hat H.
    \]
    We now have a recurrence for $\| \Li \|$ in terms of $ \| \Le_{t_{i-1}}^{k_{i-1}} \| $. Note that since $\hat \Pre_i$, $\hat \Pre_i'$, and $\Pri$ all decay to zero, this is an exponentially decaying recurrence with an additional term that also decays to zero, and hence $\| \Li \| \to 0$.
\end{proof}

Note that Lemma~\ref{lem:subseqconverge} only shows that the norm of the linear cost coefficient $\| \Li \| \to 0$ for a subsequence of the indices that would be observed for a full trajectory. We now need to include the intermediate time steps where a disturbance does not incur, as illustrated in Figure~\ref{fig:distsublin}. The following theorem makes this rigorous and presents the main result of this section.

\begin{theorem} \label{thm:distawareconv}
    Under the conditions of Lemma~\ref{lem:subseqconverge}, we relax strict monotonicity to include intermediate time steps and obtain our desired sequences of the form
    \[
        (t_j', k_j') := (t_0, k_0), (t_0-1, k_0), \dots, (t_{1}+1, k_0), (t_{1}, k_{1}), \dots
    \]
    where with some abuse of notation $(t_j', k_j')$ refers to the $j$th pair from the right hand side sequence, with $j \in \Zp$. If $k_i / t_i \to 0$ as $i \to \infty$, then $\Big\| \Lj \Big\| \to 0$ and the disturbance-aware policy converges to the blind online policy.
\end{theorem}
\begin{proof}
    Lemma~\ref{lem:subseqconverge} gives us that $\| \Li \| \to 0$. Consider an intermediate term $\| \Le_{t_i-t'}^{t_i} \|$ where $t' < \bar t_i$. Using a similar strategy and notation to the proof of Lemma~\ref{lem:subseqconverge}, we can note that
    \[
        \| \Le_{t_i-t'}^{t_i} \| \leq (1 - \gamma)^{t'} \| \Li \| + \hat \Pre_i \cdot \frac{\hat r + \wb \hat H}{\gamma} 
    \]
    Since $(1 - \gamma)^{t'} < 1$ and $\hat \Pre_i \to 0$, we can upper bound the norm of intermediate terms by the norm of the previous $\| \Li \|$ and an additional term that decays to zero as $k_i / t_i \to 0$. Therefore, we have that 
    \[
        \Big\| \Lj \Big\| \to 0 \textrm{ as } j \to \infty.
    \]
    This holds for an arbitrary sequence where $k_i / t_i \to 0$, and therefore it is easy to show also that
    \[
        \Big\| \tilde{\Le}_{t_j'}^{k_j'}  \Big\| \to 0 \textrm{ as } j \to \infty.
    \]
    Inspection of \eqref{eqn:distawarecontrol} and noting that $\Pri \to 0$ yields
    \[
        \Big(u_{t_j'}^{k_j'}\Big)^* \to -K_{t_j'} x_{t_j'},
    \]
    i.e. the disturbance-aware policy converges to the blind online policy for the sequence $(k_j', t_j')$, completing the proof.
\end{proof}

Theorem~\ref{thm:distawareconv} conveys that as long as the number of disturbances grows sublinearly in the time horizon, the disturbance-aware policy converges to the blind online policy. This matches the natural intuition that if disturbances become more sparse, without knowing their exact locations it becomes less cost-effective to mitigate their potential impact by deviating from the blind online policy.

\section{Blind Policy Regret} \label{sec:blindregret}
We establish a finite-horizon regret bound between the blind online policy and the optimal offline policy in the presence of sparse disturbances. This is accomplished by an intermediate reduction to the cost of the blind online policy without disturbances. Assumption~\ref{ass:decay} is not needed in this section, and the intermediate stability requirement in Assumption~\ref{ass:controllability} is required for Theorem~\ref{thm:optofflinebound} but not Theorem~\ref{thm:blindonlinebound}. We revert the ordering of disturbances $d_{t_k} = \w$ to the natural chronological order $t_1 \leq t_2 \leq \cdots \leq t_{|D|}$.  

We let $J_0^w(x_0)$ denote the cost-to-go under the blind policy for an initial state $x_0$ with disturbances $\{\w\}$ added at times $\{t_k\}$, $k \in \{0,\dots,|D|\}$. Similarly, $J_0(x_0)$ is the cost-to-go of the blind policy \textit{without} disturbances, i.e. $d_t=0 ~ \forall t$. The cost of the optimal offline policy with the disturbance sequence is denoted as $V_0^w(x_0)$ and given by \eqref{eqn:optofflinecost}. We will bound $J_0^w(x_0) - V_0^w(x_0)$ by a reduction to $J_0(x_0)$.

\begin{lemma} \label{lem:Pb}
    Under Assumption~\ref{ass:controllability}, the Riccati sequence $\P$ is bounded such that $\| \P \| \leq \Pb$ for all $t \in \bio -\infty \isep T \eic$ for some constant $\Pb \geq 1$.
\end{lemma}
\begin{proof}
    Under the controllability and observability conditions in Assumption~\ref{ass:controllability}, it is well known that the Ricatti sequence $\P$ converges as $t \to -\infty$ \cite{bertsekas2012dynamic}. Convergence of the matrix sequence implies convergence of the norm, and convergent sequences are bounded. If this bound is less than $1$, we choose $\Pb = 1$ to ensure that $\Pb \geq 1$.
\end{proof}

\begin{remark}
    Note that the previous lemma seems stronger than necessary, as in this section we are operating in the finite-horizon domain with $t \in \bic 0 \isep T \eic$. However, Lemma~\ref{lem:Pb} is important for establishing that the bound $\Pb$ is independent of the time horizon of the system.
\end{remark}

\begin{theorem} \label{thm:blindregret}
    Consider the evolution of a trajectory starting from $x_0$ which is perturbed by $|D|$ disturbances of magnitude at most $\wb$. The regret of the blind online policy compared to the optimal offline policy is bounded by
    \begin{align*}
    |J^w_0(x_0) - V_0^w(x_0)| &\leq O(|D| \wb \|x_0\|) + O(|D|^2 \wb^2),
    \end{align*}
    where $O(\cdot)$ absorbs constants dependent on the dynamics matrices $A$, $B$ and design matrices $Q$, $Q_T$, $R$.
\end{theorem}
\begin{proof}
    Using the triangle inequality
    \begin{align*}
        |J^w_0(x_0) - V_0^w(x_0)| \leq{}&|J^w_0(x_0) - J_0(x_0)| + |J_0(x_0) - V_0^w(x_0)|.
    \end{align*}
We can now apply Theorem~\ref{thm:blindonlinebound} and Theorem~\ref{thm:optofflinebound} to derive 
    \begin{align*}
        |J^w_0(x_0) - V_0^w(x_0)| &\leq 2 |D| \wb \Pb \left(2 \|x_0\| + \wb \right) \\
                                  &{}+ |D|^2 \wb^2 (2\Pb^2 + 3 \Pb) \\
              &{}+ \frac{|D|^2 \wb \Pb (3 \Pb \wb + \gamma^{-2}) \|B\|^2}{\eigmin{R}},
    \end{align*}
    where $\gamma$ is as in Assumption~\ref{ass:controllability}, $R$ is the control cost matrix in \eqref{eqn:cost}, and $\Pb$ is from Lemma~\ref{lem:Pb}. This yields the desired expression as $\gamma$ and $\Pb$ are given by $A$, $B$, $Q$, $Q_T$, and $R$.
\end{proof}

The main utility of Theorem~\ref{thm:blindregret} is that the asyptotic suboptimality of the blind online policy can be upper bounded by a quadratic function of $|D|$ and $\wb$. While the quadratic nature of the cost \eqref{eqn:cost} hints at the result, this is nontrivial as the optimal offline policy is counterfactual in that it assumes complete knowledge of all future disturbances. Note that the derived bound has no dependence on the time horizon $T$. We derive the result by reducing to the cost $J_0(x_0)$ of the blind online policy with no disturbances, and in doing so establish Theorems~\ref{thm:blindonlinebound} and \ref{thm:optofflinebound}, which are also interesting standalone results.

\subsection{Blind policy additional cost bound}
Consider the Riccati sequence $\P[0], \dots, \P[T]$, with corresponding feedback matrices $\K$ as in \eqref{eqn:feedback}. The cost-to-go of an initial state $x_0$ under the blind policy with no disturbances is $J_0(x_0) = x_0\T \P[0] x_0$. We will bound the additional cost incurred by $|D|$ disturbances of magnitude at most $\wb$ under the blind policy $u_t = -\K x_t$.

\begin{theorem} \label{thm:blindonlinebound}
    Consider the blind online LQR policy for a controllable and observable system which is perturbed by $|D|$ disturbances of bounded norm $\wb$ over a finite time horizon. The incurred cost difference compared to the unperturbed trajectory starting at the same initial state $x_0$ is bounded by
    \begin{align*}
    |J^w_0(x_0) - J_0(x_0)| \leq 2 |D| \wb \Pb \left(\|x_0\| + \wb + |D| \wb \Pb \right),
    \end{align*}
    where $\Pb$ is from Lemma~\ref{lem:Pb}.
\end{theorem}
\begin{proof}
    Consider the perturbed trajectory $x_t$, with disturbances $w_1, \dots, w_{|K|}$ being added at time steps $t_1 \leq \dots \leq t_{|K|}$. For notational simplicity, we let $t_k' = t_k+1$ and $y_{t_k'} = x_{t_k'} - \w$ with $y_0=x_0$; i.e., $y_{t_k'}$ represents what the state would have been at time step $t_k+1$ without the disturbance $\w$ at the previous time step. Using the cost-to-go, we can write the additional cost of the perturbed trajectory as
\begin{align}
    |J^w(x_0) - J(x_0)| \nonumber &= \Big| \sum_{t_k \in D} \left[ -\ykp\T \Pkp \ykp + \xkp\T \Pkp \xkp \right] \Big|  \nonumber \\
                                &= \Big| \sum_{t_k \in D} \left[ 2 \ykp\T \Pkp \w + \w\T \Pkp \w \right] \Big| \nonumber \\
                                &\leq 2 |D| \wb^2 \Pb + 2 \wb\sum_{t_k \in D} \| \ykp\T \Pkp \|. \label{eqn:blindbound1}
\end{align}
We now wish to upper bound $\| \ykp\T \Pkp \|$. We first establish upper bounds $ \ykp\T \Pkp \ykp \leq b_k$ for a recurrence $b_k$. Proceeding by induction, let $b_0 = x_0\T \P[0] x_0$. The induction hypothesis states
\[
    \ykp\T \Pkp \ykp \leq b_k,
\]
which also implies that
\begin{align*}
    \| \ykp\T \Pkp \| &\leq \sqrt{\Pb} \cdot \sqrt{\ykp\T \Pkp \ykp} = \sqrt{\Pb b_k}.
\end{align*}
Therefore, with some manipulation we have that
\[
    \xkp\T \Pkp \xkp = (\ykp + \w)\T \Pkp (\ykp + \w) \leq b_k + 2 \wb \Pb^{3/2} \sqrt{b_k} + \wb^2 \Pb.
\]
Since no disturbances act between $\xkp$ and $y_{t_{k+1}'}$, by optimal control the cost-to-go must decrease and we have
\begin{align*}
    y_{t_{k+1}'}\T \P[t_{k+1}'] y_{t_{k+1}'} \leq \xkp\T \Pkp \xkp \leq b_k + 2 \wb \Pb^{3/2} \sqrt{b_k} + \wb^2 \Pb \defeq b_{k+1}.
\end{align*}
This is a nonlinear recurrence, for which closed-form solutions are difficult to find. We make the simplification $\wb^2 \Pb \leq \wb^2 \Pb^3$ since $\Pb \geq 1$ by Lemma~\ref{lem:Pb}; this is a small concession to our bound since $\Pb$ is a fixed, precomputed quantity. The recurrence now reads
\begin{align*}
    b_{k+1} &= b_k + 2 \sqrt{\wb^2 \Pb^3 b_k} + \wb^2 \Pb^3 \\
            &= \left(\sqrt{b_k} + \sqrt{\wb^2 \Pb^3}\right)^2.
\end{align*}
Square rooting both sides yields
\begin{align*}
    \sqrt{b_{k}} = \sqrt{b_0} + k \sqrt{\wb^2 \Pb^3}.
\end{align*}
We can now bound our original term of interest
\begin{align*}
    \| \ykp\T \Pkp \| \leq \sqrt{\Pb} \sqrt{b_k} \leq \sqrt{\Pb} \cdot \left(\sqrt{x_0\T \P[0] x_0} + k \sqrt{\wb^2 \Pb^3}\right).
\end{align*}
Noting that this expression is maximized by $k=|D|$ and substituting into \eqref{eqn:blindbound1} completes the proof.
\end{proof}

Theorem~\ref{thm:blindonlinebound} yields an easily-computable bound for the maximum incurrable difference between the cost of a system under nominal, no-disturbance operation and the cost with the inclusion of sparse, unexpected disturbances. The bound scales quadratically in both $|D|$ and $\wb$, while the dependence on the initial state magnitude $\|x_0\|$ is merely linear.

\subsection{Optimal offline policy cost bound}
Here we desire to bound the difference between the cost of the blind policy without disturbances to the optimal online policy with disturbances.

\begin{theorem} \label{thm:optofflinebound} 
    Under Assumption~\ref{ass:controllability}, consider $|D|$ disturbances over a finite time horizon, each with norm bounded by $\wb$. The cost difference between the optimal online policy without disturbances and the optimal offline policy with disturbances is bounded by
    \begin{align*}
        &|J_0(x_0) - V_0^w(x_0)| \\
        \leq{}&2 \|x_0\| |D| \Pb \wb + |D|^2 \wb \Pb \left( 3\wb + \frac{(3 \Pb \wb + \gamma^{-2}) \|B\|^2}{\eigmin{R}} \right),
    \end{align*}
    where $\gamma$ is as in Assumption~\ref{ass:controllability}, $R$ is the control cost matrix in \eqref{eqn:cost}, and $\Pb$ is from Lemma~\ref{lem:Pb}.
\end{theorem}
\begin{proof}
    Since $J_0(x_0) = x_0\T \P[0] x_0$, we have
    \begin{align}
        |J_0(x_0) - V_0^w(x_0)| \leq |v_0\T x_0| + |q_0|. \label{eqn:optofflinestart}
    \end{align}
    We will bound each of these terms individually.

    \textit{Bounding} $|v_0\T x_0|$: The structure of \eqref{eqn:optofflinevector} yields
    \begin{align*}
        v_0 = \sum_{t_k \in D} \Big( \prod_{t=0}^{t_k-1} A\T S_t \Pp^{-1} \Big) 2 A\T S_{t_k} w_k.
    \end{align*}
    Note that $\Pp^{-1} S_t A = A - B \K$, therefore it follows from Assumption~\ref{ass:controllability} that $\rho(A\T S_t \Pp^{-1}) < 1 - \gamma ~ \forall t$. We now have
    \begin{align}
        |v_0\T x_0| &\leq \|x_0\| \cdot 1 \cdot \sum_{t_k \in D} \|2 A\T S_{t_k} w_k\| \nonumber \\
                  &= 2\|x_0\| \sum_{t_k \in D} \|A\T S_{t_k} \P[t_k+1]^{-1}\| \| \P[t_k+1] w_k\| \nonumber \\
                  &\leq 2 \|x_0\| |D| \Pb \wb \label{eqn:v0x0bound}.
    \end{align} 

    \textit{Bounding} $|q_0|$: The structure of the recurrence \eqref{eqn:optofflinescalar} and the observation that $d_t=0$ if $t \not \in D$ yields
    \begin{subequations}
    \begin{align}
        |q_0| \leq &\sum_{t_k \in D} \| \w\T S_{t_k+1} \w \| \label{eqn:q0bound1} \\
        + &\sum_{t_k \in D} \| v_{t_k+1}\T \P[t_k+1]^{-1} S_{t_k} \w \| \label{eqn:q0bound2} \\
        + &\sum_{t=0}^T \sfrac{1}{4} ~ \|v_{t+1}\|^2 \|B(B\T \Pp B + R)^{-1}B\T\|. \label{eqn:q0bound3}
    \end{align}
    \end{subequations}
    We will bound \eqref{eqn:q0bound1}, \eqref{eqn:q0bound2}, and \eqref{eqn:q0bound3} individually. Regarding \eqref{eqn:q0bound1}, we can write:
    \begin{align*}
        \sum_{t_k \in D} \| \w\T S_{t_k+1} \w \| &\leq |D| \wb^2 \| S_t \| \\
                                                 &\leq |D| \wb^2 \Pb \|I - B(B\T \Pp B + R)^{-1} B\T \Pp\| \\
                                                 &\leq |D| \wb^2 \Pb \left( 1 + \|B\|^2 \Pb \|(B\T \Pp B + R)^{-1}\| \right) \\
                                                 &\leq |D| \wb^2 \Pb \left( 1 + \frac{\|B\|^2 \Pb}{\eigmin{R}} \right).
    \end{align*}
    The last step follows from $R \succ 0$ and $B\T \Pp B \succeq 0$. We can now bound \eqref{eqn:q0bound2}:
    \begin{align*}
        \sum_{t_k \in D} \| v_{t_k+1}\T \P[t_k+1]^{-1} S_{t_k} \w \|  &\leq \sum_{t_k \in D} \| v_{t_k+1} \| \wb \left( 1 + \frac{\|B\|^2 \Pb}{\eigmin{R}} \right) \\
                                                                      &\leq 2 |D|^2 \wb^2 \Pb \left( 1 + \frac{\|B\|^2 \Pb}{\eigmin{R}} \right).
    \end{align*}
    By noting that the previous $\|v_0\|$ bound also holds similarly for any $\|v_t\|$. Finally, \eqref{eqn:q0bound3} can by manipulated as:
    \begin{align*}
        &\sum_{t=0}^T \|v_{t+1}\|^2 \|B(B\T \Pp B + R)^{-1}B\T\| \\
        \leq{}&\frac{\|B\|^2}{\eigmin{R}} \sum_{t=0}^T \|v_{t+1}\|^2 \\
        \leq{}&\frac{\|B\|^2}{\eigmin{R}} \Bigg( \sum_{t=0}^T \|v_{t+1}\| \Bigg)^2 \\
        \leq{}&\frac{\|B\|^2}{\eigmin{R}} \Bigg( \sum_{t=0}^T ~ \sum_{\substack{t_k \in D \\ t_k>t}} \Big\| \Big( \prod_{t'=t}^{t_k-1} A\T S_{t'} \P[t'+1]^{-1} \Big) 2 A\T S_{t_k} w_k \Big\| \Bigg)^2. \\
    \end{align*}
    Recalling that $\| A\T S_{t'} \P[t'+1]^{-1} \| \leq 1 - \gamma$ by Assumption~\ref{ass:controllability} yields
    \begin{align*}
        \leq{}&\frac{4 \|B\|^2}{\eigmin{R}} \Bigg( \sum_{t=0}^T ~ \sum_{\substack{t_k \in D \\ t_k>t}} (1 - \gamma)^{t_k - t} \Big\| A\T S_{t_k} \P[t_k+1]^{-1} \P[t_k+1] \w \Big\| \Bigg)^2 \\
        \leq{}&\frac{4 \Pb \wb \|B\|^2}{\eigmin{R}} \Bigg( \sum_{t=0}^T ~ \sum_{\substack{t_k \in D \\ t_k>t}} (1 - \gamma)^{t_k - t} \Bigg)^2 \\
        \leq{}&\frac{4 \Pb \wb \|B\|^2}{\eigmin{R}} \Bigg( \sum_{t_k \in D} ~ \sum_{t=0}^{t_k} (1 - \gamma)^{t_k - t} \Bigg)^2 \\
        \leq{}&\frac{4 \Pb \wb \|B\|^2}{\eigmin{R}} \Bigg( \frac{|D|}{1 - (1 - \gamma)} \Bigg)^2 \\
        \leq{}&\frac{4 \Pb \wb |D|^2 \|B\|^2}{\gamma^2 \eigmin{R}}.
    \end{align*}
    We combine bounds for \eqref{eqn:q0bound1}, \eqref{eqn:q0bound2}, \eqref{eqn:q0bound3}, reintroducing the factor $\sfrac{1}{4}$, to obtain
    \begin{align*}
        |q_0| &\leq (|D| + 2 |D|^2) \wb^2 \Pb \left( 1 + \frac{\|B\|^2 \Pb}{\eigmin{R}} \right) + \frac{\Pb \wb |D|^2 \|B\|^2}{\gamma^2 \eigmin{R}} \\
        &\leq 3 |D|^2 \wb^2 \Pb \left( 1 + \frac{\|B\|^2 \Pb}{\eigmin{R}} \right) + \frac{\Pb \wb |D|^2 \|B\|^2}{\gamma^2 \eigmin{R}} \\
        &\leq |D|^2 \wb \Pb \left( 3\wb + \frac{(3 \Pb \wb + \gamma^{-2}) \|B\|^2}{\eigmin{R}} \right).
    \end{align*}
    We combine this bound on $|q_0|$ with the bound \eqref{eqn:v0x0bound} for $|v_0\T x_0|$ to yield the result.
\end{proof}

\begin{figure*}[h!]
    \centering
    \begin{subfigure}[b]{0.4\linewidth}
        \includegraphics[width=\linewidth]{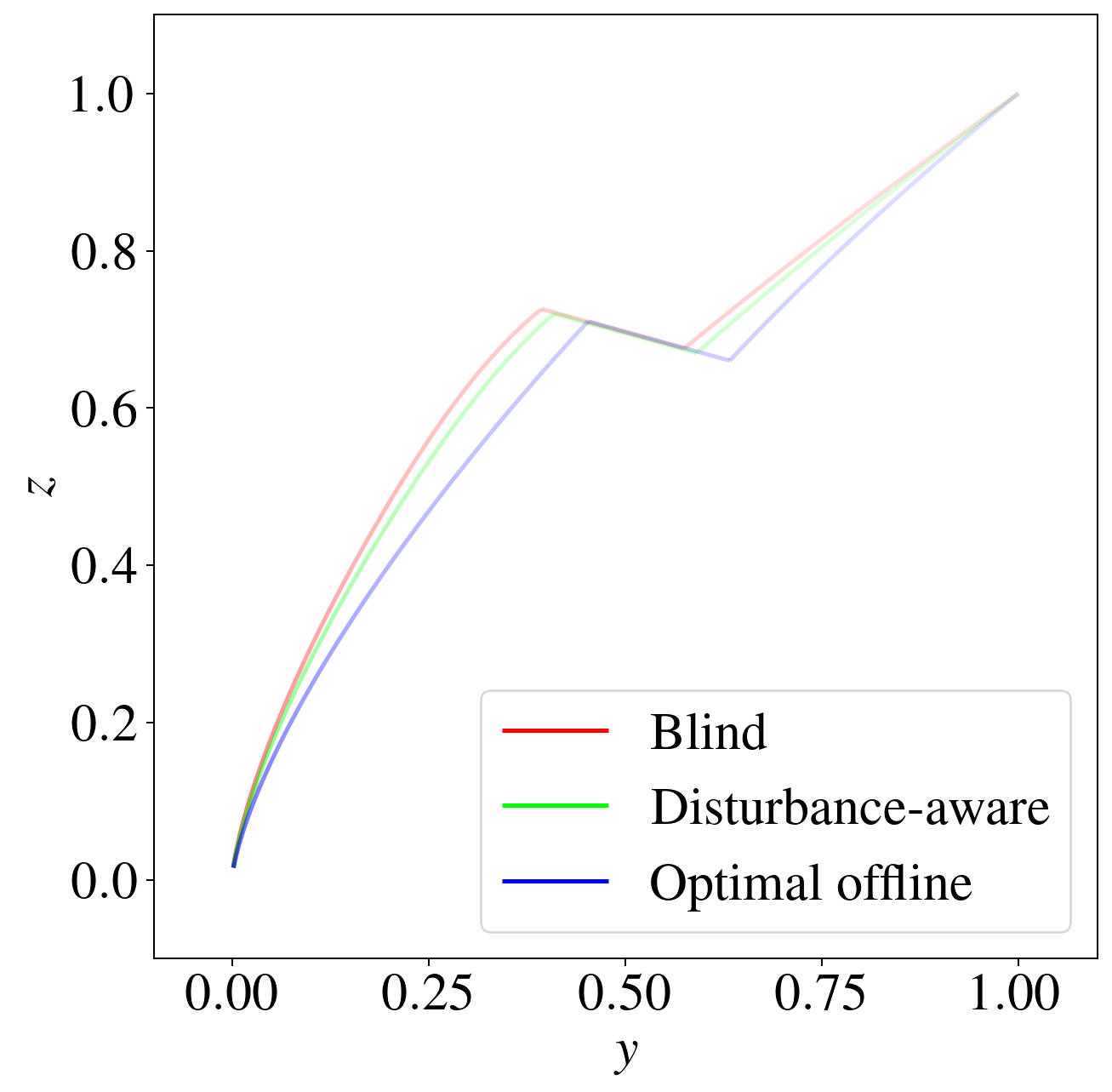}
        \phantomcaption{}
        \label{fig:trajectory}
    \end{subfigure}
    \hfill
    \begin{subfigure}[b]{0.57\linewidth}
        \includegraphics[width=\linewidth]{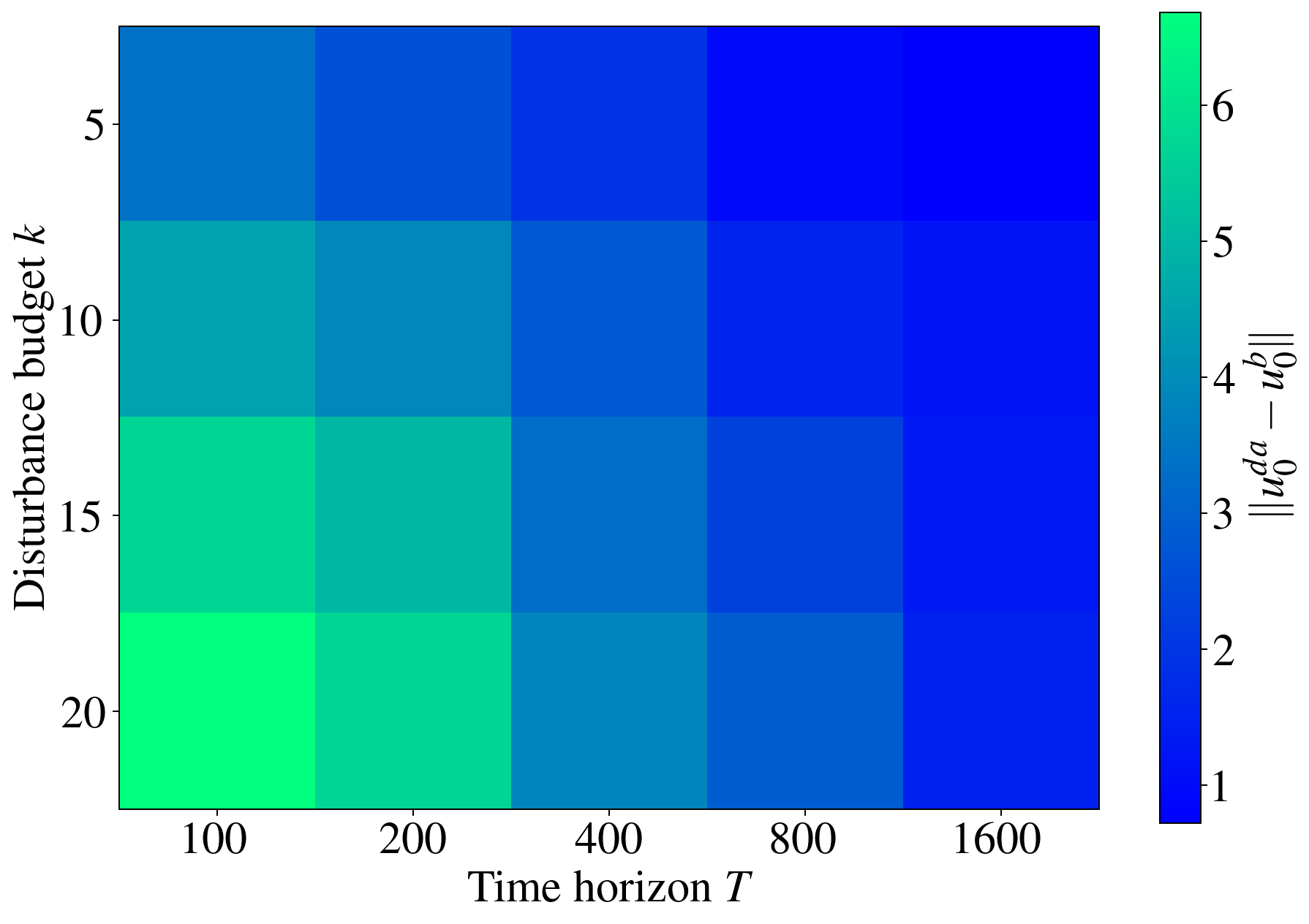}
        \phantomcaption{}
        \label{fig:sweep}
    \end{subfigure}
    \vspace{-0.5cm}
    \caption{
        (\subref{fig:trajectory}) Example trajectories for the system in Section~\ref{sec:experiments}. 
        (\subref{fig:sweep}) $\ell_2$ distance between the first action taken by the disturbance-aware policy ($u_0^{da}$) and blind policy ($u_0^b$).
    }
    \vspace{-0.5cm}
    \label{fig:main}
\end{figure*}

\section{Experiments}
\label{sec:experiments}

We illustrate our theoretical results on a planar double integrator system with the state
\[
    x = \begin{bmatrix} y & \dot y & z & \dot z \end{bmatrix} \T,
\]
for horizontal position $y$ and vertical position $z$. The dynamics and cost matrices are as follows:
\begin{align*}
    A \defeq \begin{bmatrix}
        1 & \Delta t & 0 & 0 \\
        0 & 1 & 0 & 0 \\
        0 & 0 & 1 & \Delta t \\
        0 & 0 & 0 & 1
    \end{bmatrix}, ~ 
    B \defeq \begin{bmatrix} 
        0 & 0 \\
        \Delta t & 0 \\
        0 & 0 \\
        0 & \Delta t 
    \end{bmatrix},
\end{align*}
\begin{align*}
    Q = Q_T \defeq \begin{bmatrix}
        2 & 0 & 0 & 0 \\
        0 & 10^{-3} & 0 & 0 \\
        0 & 0 & 1 & 0 \\
        0 & 0 & 0 & 10^{-3} 
        \end{bmatrix}, ~ 
    R \defeq \begin{bmatrix}
        10^{-2} & 0 \\
        0 & 10^{-2} 
    \end{bmatrix}.
\end{align*}
We evaluate the blind, disturbance-aware, and optimal offline policies from the initial condition $x_0 = [1, 0, 1, 0]$. For simplicity of visualization, we consider disturbances drawn uniformly from the sphere with radius $\wb = 0.3$ with no velocity components. The time step $\Delta t$ is $0.005$. The disturbance-aware policy is provided the uniform disturbance probability model $\Pr = k / (T - t)$.

Figure~\ref{fig:trajectory} shows example trajectories for the blind, disturbance-aware, and optimal offline policies with a single disturbance. The optimal offline policy anticipates the disturbance by preemptively steering the trajectory further downwards, while the disturbance-aware trajectory does the same but more cautiously since it does not know the precise disturbance location.
Figure~\ref{fig:sweep} shows the norm of the difference between the first action taken by the disturbance-aware policy ($u_0^{da}$) and blind policy ($u_0^b$) for various time horizons and disturbances. As expected from Theorem~\ref{thm:distawareconv}, as the time horizon increases or disturbance budget decreases, the actions of the disturbance-aware policy approach those of the blind policy.

\section{Conclusion}
This paper examines the problem of LQR control with bounded sparse perturbations. We derive the structure of the disturbance-aware policy, which knows the values of future disturbances but only has probabilistic knowledge of their locations in time. Under the condition that the number of disturbances grows sublinearly in the time horizon, we show that the disturbance-aware policy converges to the blind online policy. Section~\ref{sec:blindregret} examines the regret between the blind online policy and the optimal offline policy via a reduction to the cost of the blind online policy without disturbances. The regret bound is shown to scale quadratically in the number of disturbances and their magnitude, but only linearly in the norm of the initial state. We conclude with some synthetic experiments illustrating our findings.

\printbibliography

\end{document}